\documentclass[a4paper]{article}
	\usepackage{fullpage}
\usepackage{amsmath}
\usepackage{amssymb}
\usepackage{amsthm}
\usepackage{xspace}
\usepackage{enumitem}
\usepackage{tikz}
\usepackage[ruled]{algorithm2e}
\usetikzlibrary{patterns,decorations.pathreplacing,positioning}

\newcommand{\M}{\ensuremath{M}\xspace}
\newcommand{\N}{\ensuremath{N}\xspace}

\newcommand{\trail}{\ensuremath{T}\xspace}
\newcommand{\cycle}{\ensuremath{C}\xspace}

\newcommand{\base}{\ensuremath{B}\xspace}
\newcommand{\symdiff}{\ensuremath{\mathop{\bigtriangleup}}}

\newcommand{\degree}[2]{\ensuremath{\operatorname{\delta}_{#1}(#2)}\xspace}
\newcommand{\incident}[2]{\ensuremath{\operatorname{I}_{#1}(#2)}\xspace}

\newcommand{\uline}[1]{\overset{#1}{\ \rule[3pt]{1em}{0.5pt}\ }} %
\newcommand{\edge}{\ensuremath{e}\xspace}
\newcommand{\nat}{\ensuremath{\mathbb{N}}\xspace}

\newcommand{\dcs}{\ensuremath{\ensuremath{ab\textsc{-DCS}}}\xspace}
\newcommand{\stdcsconn}{\ensuremath{st}\text{-\textsc{DCSConn}}\xspace}

\newcommand{\atight}{$a$-tight\xspace}
\newcommand{\btight}{$b$-tight\xspace}
\newcommand{\abtight}{$ab$-tight\xspace}
\newcommand{\abfixed}{$ab$-fixed\xspace}
\newcommand{\abconstrained}{$ab$-con\-strain\-ed\xspace}
\newcommand{\yes}{\textsc{Yes}\xspace}
\newcommand{\no}{\textsc{No}\xspace}

\newcommand{\instance}{\ensuremath{\mathcal{I}}\xspace}

\newcommand{\mfixed}{\M-fixed\xspace}

\newcommand{\NotA}{\operatorname{\textsc{NotA}}\xspace}
\newcommand{\NotB}{\operatorname{\textsc{NotB}}\xspace}
\newcommand{\Even}{\operatorname{\textsc{Even}}\xspace}

\newcommand{\mfixedalgo}{\textsc{M-FixedSubgraph}\xspace}
\newcommand{\atdecomp}{\textsc{AlternatingTrailDecomposition}\xspace}

\newtheorem{theorem}{Theorem}
\newtheorem{lemma}{Lemma}

\newtheorem{definition}{Definition}
\newtheorem{proposition}{Proposition}

\newtheorem{mainthm*}{Theorem}

\title{Degree-constrained Subgraph \\Reconfiguration is in P}
\author{Moritz M\"uhlenthaler\thanks{Research funded in parts by the School of
Engineering of the University of Erlangen-Nuremberg.}}

\begin{document}
\maketitle
\begin{abstract}
The degree-constrained subgraph problem asks for a subgraph of a given graph
such that the degree of each vertex is within some specified bounds.  We study the
following reconfiguration variant of this problem: Given two solutions to a
degree-constrained subgraph instance, can we transform one solution into the
other by adding and removing individual edges, such that each intermediate
subgraph satisfies the degree constraints and contains at least a certain minimum
number of edges?  This problem is a generalization of the matching
reconfiguration problem, which is known to be in P. We show that even in the
more general setting the reconfiguration problem is in P.
\end{abstract}

\section{Introduction}

A reconfiguration problem asks whether a given solution to a combinatorial
problem can be transformed into another given solution in a step-by-step
fashion such that each intermediate solution is ``proper'', where the
definition of proper depends on the problem at hand. For instance, in the
context of vertex coloring reconfiguration, ``step-by-step'' typically means that the
color of a single vertex is changed at a time, and ``proper'' has the usual
meaning in the graph coloring context. An issue of particular interest is the
relation between the complexity of the underlying combinatorial problem and its
reconfiguration variant. This complexity relation has been studied for
classical combinatorial problems including for example graph coloring,
satisfiability, matching, and the shortest path
problem~\cite{Ito:11,Bonsma:13,Cereceda:11}. Surprisingly, the reconfiguration
variants of some tractable problems turn out to be intractable~\cite{Ito:11},
and vice versa~\cite{Wrochna:15}. An overview of recent results on
reconfiguration problems can be found in~\cite{Heuvel:13}.

In this work we investigate the complexity of the reconfiguration problem
associated with the $(a,b)$-degree-constrained subgraph (\dcs) problem. Let $G
= (V, E)$ be a graph and let  $a, b: V \rightarrow \nat$ be two functions
called \emph{degree bounds} such that for each vertex $v$ of $G$ we have
$0 \leq a(v) \leq b(v) \leq \degree{G}{v}$, where $\degree{G}{v}$ denotes the
degree of $v$ in $G$. The task is to decide if there is a subgraph $S$ of $G$
that satisfies the degree constraints, that is, for each vertex $v$ of $S$,
$\degree{S}{v}$ is required to be at least $a(v)$ and at
most $b(v)$. Typically, the intention is to find among all subgraphs of $G$
that satisfy the degree constraints one with the greatest number of edges. This
problem is a generalization of the classical maximum matching problem and can
be solved in polynomial time by a combinatorial
algorithm~\cite{Shiloach:81,Gabow:83}. The \dcs reconfiguration problem is
defined as follows:
\begin{definition}{(\stdcsconn)}
\mbox{}\newline
\emph{INSTANCE:} An \dcs instance, source and target solutions \M, \N, an integer $k \geq 1$. 
\newline
\emph{QUESTION:} Is it possible to transform \M into \N by adding/removing a
single edge in each step such that each intermediate subgraph satisfies the
degree constraints and contains at least $\min\{|E(\M)|,|E(\N)|\}-k$ edges?
\end{definition}

Our main result is the following
\begin{mainthm*}
  \stdcsconn can be solved in polynomial time.
\end{mainthm*}

It was shown by Ito et al.~in~\cite[Proposition 2]{Ito:11} that the analogous matching
reconfiguration problem can be solved in polynomial time. According to our result,
the reconfiguration problem remains tractable even in the more general \dcs
reconfiguration setting. The proof of the main result essentially contains an
algorithm that determines a suitable sequence of edge additions/removals if one
exists. The number of reconfiguration
steps is bounded by $O(|E|^2)$. The algorithm also provides a certificate for \no-instances.

\section{Notation}

In this paper we deal with subgraphs of some simple graph $G = (V,
E)$, which is provided by \dcs a problem instance. The subgraphs of concern are induced by
subsets of $E$. For notational convenience, we identify these
subgraphs with the subsets of $E$ and can therefore use standard set-theoretic
notation ($\cap$, $\cup$, $-$) for binary operations on the subgraphs. Let $H$
and $K$ be two subgraphs of $G$, denoted by $H, K \subseteq G$. By $E(H)$ we
refer explicitly to the set of edges of the graph $H$. We write $H + K$ for
the union of $H$ and $K$ if they are disjoint. By $H \symdiff K$ we denote the
\emph{symmetric difference} of $H$ and $K$, that is $H \symdiff K :=
(H - K) + (K - H)$.  If $e$ is an edge of $G$ we may write $H + e$ and $H - e$
as shorthands for $H + \{e\}$ and $H - \{e\}$, respectively. We denote the
degree of a vertex $v$ of $H$ by $\degree{H}{v}$. A walk $v_0 \uline{\edge_0}
\ldots \uline{\edge_{t-1}} v_t$ in $G$ is a \emph{trail} if $\edge_0,
\ldots, \edge_{t-1}$ are distinct. The vertices $v_0$ and $v_t$ are called
\emph{end vertices}, all other vertices are called \emph{interior}.  A trail
without an edge is called \emph{empty}. A non-empty trail is \emph{closed} if
its end vertices agree, otherwise it is \emph{open}. A closed trail is also
called a \emph{cycle}. In a slight abuse of notation we will sometimes consider
trails in $G$ simply as subgraphs of $G$ and combine them with other
subgraphs using the notation introduced above.  A trail is called \emph{$(K,
H)$-alternating} if its edges, in the order given by the trail, are
alternatingly chosen from $K - H$ and $H - K$. An odd-length $(K,H)$-alternating
trail \trail is called $K$-augmenting if $|E(K \symdiff \trail)| > |E(K)|$. 

Let $G$ be the graph and $a,b: V \rightarrow \nat$ be the degree bounds of an
\dcs instance. A subgraph $\M \subseteq G$ that satisfies the degree constraints 
is called \emph{\abconstrained}. A vertex $v$ of $\M$ is called
\emph{\atight} (\emph{\btight}) \emph{in \M} if $\degree{M}{v} = a(v)$
($\degree{\M}{v} = b(v)$). A vertex is called \emph{\abfixed in $\M$} if it is
both \atight and \btight in $M$.  We say that $\M$ is \atight (\btight) if each
vertex of $\M$ is \atight (\btight).  A closed $(\M,\N)$-alternating trail
$\trail = v_0 \uline{} \ldots \uline{} v_t$ of even length is called
\emph{alternatingly \abtight} in \M if for each $i$, $0 \leq i \leq t$, $v_i$
is \atight iff $i$ is even and \btight iff $i$ is odd, or vice versa.

\section{$ab$-Constrained Subgraph Reconfiguration}

Throughout this section, we assume that we are given some \stdcsconn instance 
$(G, \M, \N, a, b, k)$, where $G = (V, E)$ is a graph, $a, b: V \rightarrow \nat$
are degree bounds, $\M, \N \subseteq G$ are \abconstrained, and $k \geq 1$. A
\emph{reconfiguration step}, or \emph{move}, adds/removes an edge to/from a
subgraph.  Given an \M-edge $e$ and an \N-edge $e'$, an \emph{elementary move}
on \M yields $\M - e + e'$, either by adding $e'$ after removing $e$ or vice
versa.  \M is \emph{$k$-reconfigurable} to \N if there is a sequence of
reconfiguration moves that transforms \M into \N such that each intermediate
subgraph respects the degree constraints and contains at least $\min\{|E(\M)|,
|E(\N)|\} - k$ edges.  Clearly, if \M is $k$-reconfigurable to \N then \M is also
$k'$-reconfigurable to \N for any $k' > k$.  \M is \emph{internally
$k$-reconfigurable} to \N if it is $k$-reconfigurable under the additional
restriction that each intermediate subgraph is contained in $\M \symdiff \N$.
If \M is not internally $k$-reconfigurable to \N but still $k$-reconfigurable
to \N then we say that \M is externally $k$-reconfigurable to \N.

The general procedure for deciding if \M is $k$-reconfigurable to \N is the
following: First, we check for the presence of obstructions that render a
reconfiguration impossible. If it turns out that reconfiguration is still
possible we reconfigure $(\M,\N)$-alternating trails in $\M \symdiff \N$, one
by one, until we either finish successfully or we obtain a certificate for
\M not being $k$-reconfigurable to \N. Curiously, it turns out that if \M is not
2-reconfigurable to \N then \M is not $k$-reconfigurable to \N for any $k \geq
2$.

\subsection{Obstructions}
\label{sec:obstructions}

When transforming \M into \N, certain parts of \M may be ``fixed'' and
therefore make a proper reconfiguration impossible.  Similar obstructions occur
for example in vertex coloring reconfiguration, where certain vertices are
fixed in the sense that their color cannot be changed~\cite{Cereceda:11}. In
our case we identify a certain subgraph of $G$ which depends on \M and \N and
cannot be changed at all.

Let $v$ be an \abfixed vertex of $G$. When
reconfiguring $M$ to $N$, no $M$-edge incident to an \abfixed vertex can be removed
and no $(G-M)$-edge incident to an \abfixed vertex can be added during a
reconfiguration process without violating the degree constraints.  Hence an
edge is fixed if it is incident to an \abfixed vertex. 
However, we may identify larger parts of $G$ that are fixed due to the given
subgraph $M$ and the degree bounds. 
If we consider each \abfixed edge to be
\mfixed, then we can identify further \M-fixed edges based on the following
observations: First if a vertex $v$ is incident to exactly $b(v)$ \M-edges and
each of them is \mfixed then all edges incident to $v$ are \mfixed.  Similarly,
if $v$ is incident to exactly $a(v)$ $M$-edges and each $(G-M)$-edge incident
to $v$ is \mfixed then each edge incident to $v$ is \mfixed.
Algorithm~\ref{alg:mfixed} shows how to identify an \mfixed subgraph of $G$
based on these observations. By $\incident{M}{v}$ we denote the set of \M-edges
incident to the vertex $v$. Some bookkeeping could be employed to speed things
up, but it is not necessary for our argument.

\begin{algorithm}
\caption{\mfixedalgo}
\label{alg:mfixed}
\DontPrintSemicolon
\SetKwInOut{Input}{input}
\SetKwInOut{InOut}{in/out}
\SetKwInOut{Output}{output}
\SetKwData{KLIST}{$K$}
\Input{\dcs instance $(G, a, b)$, \abconstrained $M \subseteq G$}
\Output{$M$-fixed subgraph $F \subseteq G$}
\SetKwData{FF}{$F$}
\SetKwData{GG}{$F'$}
$\FF \longleftarrow \emptyset;\qquad \GG \longleftarrow \{ \edge \in G \mid \edge \text{ is fixed} \}$\;
\While{ $|\GG| > |\FF|$ }
{
  $\FF \longleftarrow \GG$\;
  \uIf{$\degree{M}{v} = b(v)$ and $\incident{M}{v} \subseteq \FF$}{$\GG \longleftarrow \GG \cup \incident{G}{v}$}
  \uIf{$\degree{M}{v} = a(v)$ and $\incident{G-M}{v} \subseteq \FF$}{$\GG \longleftarrow \GG \cup \incident{G}{v}$}
}
\Return \FF\;
\end{algorithm}

\begin{proposition}
  Let $M, N$ be \abconstrained subgraphs of $G$ and let $F
  \subseteq G$ be $M$-fixed. If $(M \symdiff N) \cap F$ is non-empty then \M is
  not $k$-reconfigurable to $N$ for any $k \geq 1$. 
  \label{prop:mfixed}
\end{proposition}

That is, any \M-fixed edge in $(M \symdiff N)$ is a \no-certificate.
As a consequence, we can check if $M$ and $N$ agree on the subgraph $F
\subseteq G$ found by Algorithm~\ref{alg:mfixed} as a preprocessing step. At
this point in particular, but also later on it will be convenient to consider
\emph{subinstances} of a given \stdcsconn instance $\instance = (G, M, N, a,
b,k)$. If $H \subseteq G$ then the corresponding subinstance $\instance_H$ is the
instance $(H, M \cap H , N \cap H, a_H, b_H,k)$, where
\begin{eqnarray*}
  a_H(v) & = & \max \{0, a(v) - \degree{(G-H) \cap M}{v} \}\\
  b_H(v) & = & b(v) - \degree{(G-H) \cap M}{v}\enspace.
\end{eqnarray*}
\begin{proposition}
  If $(\M \symdiff \N) \cap F = \emptyset$ then \instance is a \yes-instance if and only if
  $\instance_{G-F}$ is a \yes instance.
  \label{prop:subinstance}
\end{proposition}

The graph $G-F$ does not have any fixed vertices and hence the $(M-F)$-fixed
subgraph produced by Algorithm~\ref{alg:mfixed} is empty.  Removing the \mfixed
subgraph of $G$ in a preprocessing step will considerably simplify our
arguments later on. It should be immediate that no fixed edges can be
introduced by reconfiguring an \abconstrained subgraph.

\subsection{Internal Alternating Trail Reconfiguration}

The next Lemma is our fundamental tool for reconfiguring alternating trails in
$\M \symdiff \N$. For any such trail \trail, it provides necessary and
sufficient conditions for $\trail \cap \M$ being internally 1-reconfigurable to
$\trail \cap \N$ by performing only elementary moves. Behind the scenes, \trail is recursively divided into
subtrails which need to be reconfigured in a certain order. If successful, the
reconfiguration procedure performs exactly $|E(\trail)|$ edge additions/removals.

\begin{lemma}
\label{lemma:em1reconfig}
Let $\trail = v_0 \uline{\edge_0} \ldots \uline{\edge_{t-1}} v_t$ be a $(\M, \N)$-alternating
trail of even length in $\M \symdiff \N$. Then $\trail \cap \M$ is internally
1-reconfigurable to $\trail \cap \N$ using only elementary moves if
and only if \trail it satisfies each of the following conditions:
\begin{enumerate}
	\item \trail contains no \abfixed vertex.\label{itm:reconf:fixed}
	\item If \trail is open then $v_0$ is not \atight and $v_t$ is not \btight in \M.\label{itm:reconf:open}
	\item If \trail is closed then each of the following is true:\label{itm:reconf:closed} 
		\begin{enumerate}
			\item \trail is not \btight in \M \label{itm:reconf:btight}
			\item \trail is not \atight in \M \label{itm:reconf:atight}
			\item \trail is not alternatingly \abtight in \M\label{itm:reconf:abtight}
		\end{enumerate}
\end{enumerate}
\end{lemma}

\begin{proof}
  Without loss of generality, let $\edge_0$ be an \M-edge.
  We first show the necessity of
  conditions~\ref{itm:reconf:fixed}--\ref{itm:reconf:closed}.  By
  Proposition~\ref{prop:mfixed}, if $\trail$ contains an $\M$-fixed vertex then
  $\trail \cap \M$  is not $k$-reconfigurable to $\trail \cap \N$ for any $k
  \geq 1$.  If \trail is open and $v_0$ is \atight in \M, then $\edge_0$ cannot be
  removed from $\trail \cap \M$
  without violating the degree constraints.  Likewise, if \trail is open and $v_t$ is 
  \btight then $\edge_t$ cannot be added to $\trail \cap \M$. If \trail is
  closed and \btight in \M, i.e.,~\ref{itm:reconf:btight} is violated, then no
  \N-edge can be added after removing any \M-edge. Similarly, if \trail is
  closed and \atight in \M, i.e.,~\ref{itm:reconf:atight} is violated, then no
  \M-edge can be removed after adding a single \N-edge. If \trail is closed and
  alternatingly \abtight then no edge can be added to or removed from $\trail
  \cap \M$, so it cannot be internally 1-reconfigurable to $\trail \cap \N$. In
  summary, if any of the
  conditions~\ref{itm:reconf:fixed}--\ref{itm:reconf:closed} is violated then
  $\trail \cap \M$ is not internally 1-reconfigurable to $\trail \cap \N$ using
  elementary moves. 

  In order to show the sufficiency of
  conditions~\ref{itm:reconf:fixed}--\ref{itm:reconf:closed} we employ the
  following general strategy: We partition \trail into $(\M,\N)$-alternating
  subtrails $R$, $Q$, $S$, each of even length and at least two of them
  non-empty.  We show that for an appropriate choice of these subtrails there
  is an ordering, say $Q$, $R$, $S$, such that the first two subtrails are
  non-empty and $Q$ satisfies
  conditions~\ref{itm:reconf:fixed}--\ref{itm:reconf:closed} in \M, $R$
  satisfies the same conditions in $\M - (Q \cap \M) + (Q \cap \N)$, and $S$,
  if non-empty, in turn satisfies the conditions in $\M - (Q \cap \M) + (Q \cap
  \N) - (R \cap \M) + (R \cap \N)$. Therefore, each non-empty subtrail can be
  dealt with individually in a recursive fashion as long as the ordering is
  respected.  The base case of the recursion consists of an
  $(\M,\N)$-alternating trail $\base = u \uline{} v \uline{} w$ of length two.
  We show that if \base
  satisfies~\ref{itm:reconf:fixed}--\ref{itm:reconf:closed} then $\base \cap
  \M$ is internally 1-reconfigurable to $\base \cap \N$ by an elementary move.
  Since \base satisfies conditions~\ref{itm:reconf:fixed} and~\ref{itm:reconf:open},
  $u$ is not \atight and $w$ is not \btight. If $v$ is \btight we can remove $u
  \uline{} v$ from \M and add $v \uline{} w$ to $\M - (u \uline{} v)$ without
  violating the degree constraints in any step. Similarly, if $v$ is not
  \btight we can add $v \uline{} w$ to \M and afterwards remove $u \uline{} v$
  from $\M + (v \uline{} w)$ without violating the degree constraints. So
  $B \cap \M$ is internally 1-reconfigurable to $B \cap \N$ by an elementary
  move as required. For the general recursion, we consider two main cases:
  \trail is either open or closed. 

  We first assume that \trail is open. Since \trail satisfies
  conditions~\ref{itm:reconf:fixed} and~\ref{itm:reconf:open}, there is some $i$, $0 \leq i < t-1$, $i$
  even, such that $v_i$ is not \atight and $v_{i+2}$ is not \btight in \M. To
  see this, assume that there is no such $i$. Then, by induction, for each $0
  \leq i < t-1$, $i$ even, $v_{i+2}$ must \btight because $v_{i}$ is not
  \atight.  However, by condition~\ref{itm:reconf:open}, $v_t$ is not \btight in
  \M, a contradiction. We pick $R$, $Q$, and $S$ as follows
  \[
	  \underbrace{v_0 \uline{e_0} v_1 \uline{e_1} \ldots \uline{e_{i-1}} v_i}_{R},\quad \underbrace{v_i \uline{e_i} v_{i+1} \uline{e_{i+1}} v_{i+2}}_{Q},\quad \underbrace{v_{i+2} \uline{e_{i+2}} v_{i+3} \uline{e_{i+3}} \ldots \uline{e_{t-1}} v_t}_{S}\enspace.
  \]
  Note that $Q$ is an open trail satisfying
  conditions~\ref{itm:reconf:fixed}--\ref{itm:reconf:closed} and $R$, $S$, if
  non-empty, can be open or closed. At this point $Q \cap \M$ is internally
  1-reconfigurable to $Q \cap \N$ as described above, and the result is $\M -
  \edge_{i} + \edge_{i+1}$. It is readily verified that if $R$ and $S$ are open
  then they satisfy conditions~\ref{itm:reconf:fixed}--\ref{itm:reconf:closed}
  in $\M - \edge_{i} + \edge_{i+1}$ and can therefore be treated independently
  after reconfiguring $Q$. However, at least one of $R$, $S$ being closed leads
  to a slight complication. 

  Let us assume that $S$ is closed. Note that, by assumption, $v_{i+2}$ is not
  \btight in \M and therefore it cannot be \atight in $\M - \edge_i + \edge_{i+1}$. Thus, if any of the
  conditions~\ref{itm:reconf:fixed}--\ref{itm:reconf:closed} is violated in \M
  then it cannot be violated in $\M - \edge_i + \edge_{i+1}$.
  Since~\ref{itm:reconf:btight}--\ref{itm:reconf:abtight} cannot be violated
  simultaneously we conclude that $S$ satisfies
  conditions~\ref{itm:reconf:fixed}--\ref{itm:reconf:closed} either in \M or
  in $\M - \edge_i + \edge_{i+1}$. An analogous argument shows that $R$
  satisfies conditions~\ref{itm:reconf:fixed}--\ref{itm:reconf:closed} either
  in \M or in $\M - \edge_i + \edge_{i+1}$.
  Therefore, there is an ordering of $R$, $Q$, $S$ that is consistent with the
  general strategy outlined above and depends on the tightness of the vertices
  $v_i$ and $v_{i+2}$ in \M.
  As a visual aid, Figure~\ref{fig:openreconf} shows a
  example of an open $(\M,\N)$-alternating trail of even length, where a proper
  choice of $R$, $Q$, and $S$ causes $R$ and $S$ to be both closed. In the shown
  example the solid edges belong to \M and the dashed edges to \N. 

  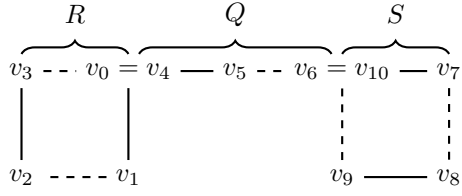
\begin{figure}
	\begin{center}
	\begin{tikzpicture}[node distance=4em,vertex/.style={},m0edge/.style={thick,solid},m1edge/.style={thick,dashed}]
		\node[]	(v0)	{$v_0=v_4$};
		\node[below of=v0]	(v1)	{$v_1$};
		\node[left of=v1]	(v2)	{$v_2$};
		\node[above of=v2]	(v3)	{$v_3$};
		\node[right of=v0]	(v5)	{$v_5$};
		\node[right of=v5]	(v6)	{$v_6=v_{10}$};
		\node[right of=v6]	(v7)	{$v_7$};
		\node[below of=v7]	(v8)	{$v_8$};
		\node[left of=v8]	(v9)	{$v_9$};
		\draw[m0edge] (v0) -- (v1);
		\draw[m0edge] (v2) -- (v3);
		\draw[m0edge] (v0) -- (v5);
		\draw[m0edge] (v6) -- (v7);
		\draw[m0edge] (v8) -- (v9);
		\draw[m1edge] (v1) -- (v2);
		\draw[m1edge] (v3) -- (v0);
		\draw[m1edge] (v5) -- (v6);
		\draw[m1edge] (v7) -- (v8);
		\draw[m1edge] (v9) -- (v6);

		\draw [thick,decoration={brace,amplitude=0.5em},decorate] ([xshift=4pt]v0.north) -- ([xshift=-4pt]v6.north) node[midway,yshift=1.5em] {$Q$};
		\draw [thick,decoration={brace,amplitude=0.5em},decorate] (v3.north) -- (v0.north) node[midway,yshift=1.5em] {$R$};
		\draw [thick,decoration={brace,amplitude=0.5em},decorate] (v6.north) -- (v7.north) node[midway,yshift=1.5em] {$S$};
	\end{tikzpicture}
	\end{center}
	\caption{Decomposition of an $(\M,\N)$-alternating trail $\trail = v_0
		\protect\uline{} \ldots \protect\uline{} v_{10}$ of even length into three trails
			$R$, $Q$, and $S$. Note that \trail is open, but $R$ and $S$ are
			closed.\label{fig:openreconf}}
  \end{figure}

  It remains to be shown that if \trail is closed and \trail satisfies
  properties~\ref{itm:reconf:fixed}--\ref{itm:reconf:closed} then $\trail \cap
  \M$ is internally 1-reconfigurable to $\trail \cap \N$ using only elementary
  moves.
  For this purpose we find a partition of \trail into subtrails that is
  compatible with our general strategy above.
  In particular, we show that if \trail is closed and satisfies
  properties~\ref{itm:reconf:fixed}--\ref{itm:reconf:closed} then \trail can
  be partitioned into two non-empty open trails $R$ and $Q$, both of even
  length,   such that $Q$
  satisfies~\ref{itm:reconf:fixed}--\ref{itm:reconf:closed}. Furthermore, 
  we show that $R$ satisfies
  conditions~\ref{itm:reconf:fixed}--\ref{itm:reconf:closed} in $\M - (Q \cap
  \M) + (Q \cap \N)$. That is, $R$ can be dealt with after
  reconfiguring $Q \cap \M$ to $Q \cap \N$. We pick any two vertices of \trail
  that are connected by an \M-edge, say $v_0$ and $v_1$, and consider the
  following cases:
  \begin{enumerate}[label=(\roman*)]
	  \item $v_0$ is neither \atight nor \btight, or $v_1$ is neither \atight nor \btight \label{itm:ctreconf:nabtight}
	  \item $v_0$ is \atight and $v_1$ is \btight, or $v_0$ is \btight and $v_1$ is \atight \label{itm:ctreconf:abtight}
	  \item $v_0$ and $v_1$ are both \btight \label{itm:ctreconf:btight}
	  \item $v_0$ and $v_1$ are both \atight \label{itm:ctreconf:atight}
  \end{enumerate}

\paragraph{Case~\ref{itm:ctreconf:nabtight}}{
  We assume without loss of generality, that $v_0$ is neither \atight nor
  \btight, since if $v_0$ is \atight or \btight, then $v_1$ must be neither \atight
  nor \btight and if this is the case we can rearrange the vertices of $T$ in the
  following way
  \begin{equation}\label{eq:rearrange}
	  T = v_1 \uline{} v_0 \uline{} v_{t-1} \uline{} \ldots \uline{} v_2 \uline{} v_1
  \end{equation}
  and establish that $v_0$ is neither \atight nor \btight. Now we choose $R$ and
  $Q$ as follows:
  \[
	  R = v_2 \uline{} v_3 \uline{} \ldots \uline{} v_{t},\quad Q = v_0 \uline{} v_1 \uline{} v_2
  \]
  If $v_{2}$ is not \btight, then $Q$
  satisfies~\ref{itm:reconf:fixed}--\ref{itm:reconf:closed} and $Q \cap \M$ can be
  reconfigured instantly to $Q \cap \N$ by an elementary move. That is, we obtain
  $\M - \edge_0 + \edge_1$ without violating the degree constraints. Now
  $v_0$ cannot be \btight and $v_{2}$
  cannot be \atight in $\M - \edge_0 + \edge_1$. Therefore, $R$ now
  satisfies~\ref{itm:reconf:fixed}--\ref{itm:reconf:closed} and can be
  reconfigured as shown in the first main case of the proof. If $v_{2}$ is
  \btight in \M, then, by analogous considerations, $R$
  satisfies~\ref{itm:reconf:fixed}--\ref{itm:reconf:closed} and $Q$
  satisfies conditions~\ref{itm:reconf:fixed}--\ref{itm:reconf:closed} in $\M -
  (R \cap \M) + (R \cap \N)$.
  }
\paragraph{Case~\ref{itm:ctreconf:abtight}}{
  Without loss of generality, we assume that $v_0$ is \atight and $v_1$ is
  \btight, since if not, we can rearrange the vertices of \trail according to
  Eq.~\eqref{eq:rearrange}. Due to property~\ref{itm:reconf:abtight}, \trail is
  not alternatingly \abtight, so there is some $i$, $0 \leq i < t$, $i$ even,
  such that $v_i$ is not \atight or $v_{i+1}$ is not \btight. If $v_i$ is not
  \atight, we choose $R$ and $Q$ to be
  \[
	  Q = v_i \uline{} v_{i+1} \uline{} \ldots \uline{} v_t (= v_0),\quad R = v_0 \uline{} v_1 \uline{} \ldots \uline{} v_i
  \]
  Otherwise, $v_{i+1}$ is not \btight and we pick $R$ and $Q$ as follows
  \[
	  Q = v_1 \uline{} v_0 \uline{} \ldots \uline{} v_{i+1},\quad R = v_{i+1} \uline{} v_{i} \uline{} \ldots \uline{} v_1
  \]
  Either way, $Q$ satisfies satisfies
  conditions~\ref{itm:reconf:fixed}--\ref{itm:reconf:closed} in \M and $R$ satisfies
  the same conditions in $\M - (Q \cap \M) + (Q \cap \N)$.
  }
\paragraph{Case~\ref{itm:ctreconf:btight}}{
  If $v_0$ and $v_1$ are both \btight then, by property~\ref{itm:reconf:btight},
  there is some $i$, $0 \leq i < t$, such that $v_i$ is not \btight. Without loss
  of generality, we assume that $i$ is even. We pick $Q$
  and $R$ as follows:
  \[
	  Q = v_0 \uline{} v_1 \uline{} \ldots \uline{} v_i,\quad R = v_i \uline{} v_{i-1} \uline{} \ldots \uline{} v_0
  \] 
  Then $Q$ satisfies~\ref{itm:reconf:fixed}--\ref{itm:reconf:closed} and $R$
  satisfies the same conditions in $\M - (Q \cap \M) + (Q \cap \N)$.
  }
\paragraph{Case~\ref{itm:ctreconf:atight}}{
  This case is analogous to case~\ref{itm:ctreconf:btight}. By
  property~\ref{itm:reconf:atight}, there is some $i$, $0 \leq i < t$, such that
  $v_i$ is not \atight. Again, without loss of generality, we assume that $i$
  is even.  We pick $Q$ and $R$ as follows
  \[
	  Q = v_0 \uline{} v_1 \uline{} \ldots \uline{} v_i,\quad R = v_i \uline{} v_{i-1} \uline{} \ldots \uline{} v_0
  \]
  and conclude that $Q$
  satisfies~\ref{itm:reconf:fixed}--\ref{itm:reconf:closed} in \M and $R$
  satisfies~\ref{itm:reconf:fixed}--\ref{itm:reconf:closed} in $\M - (Q \cap \M)
  + (Q \cap \N)$.
  }

From our consideration of the various cases we conclude that a given
$(\M,\N)$-alternating trail \trail that satisfies
conditions~\ref{itm:reconf:fixed}--\ref{itm:reconf:closed} can be recursively
partitioned into subtrails as outlined in the general
strategy above.  Since the single base case of the recursion employs only
elementary moves on edges of \trail, $\trail \cap \M$ is internally
1-reconfigurable to $\trail \cap \N$ using only elementary moves.
\end{proof}

An $(\M,\N)$-alternating trail is \emph{maximal} if there is no suitable edge
in $\M \symdiff \N$ to extend the trail at one of its end nodes.  The
subsequent lemmas establish sufficient conditions for maximal alternating
trails to be internally 1- or 2-reconfigurable. Such trails are important in
the proof of Theorem~\ref{thm:stdcsconn}. Note however, that in contrast to
Lemma~\ref{lemma:em1reconfig} we are not restricted to elementary moves. 
In the following, let $\trail = v_0 \uline{\edge_0} \ldots \uline{\edge_{t-1}}
v_t$ be a maximal $(\M,\N)$-alternating trail in $\M \symdiff \N$ such that no
vertex of \trail is \abfixed.

\begin{lemma}
	If \trail is open and has even length then $\trail \cap \M$ is internally
	1-recon\-figurable to $\trail \cap \N$.
	\label{lemma:even1reconfig}
\end{lemma}
\begin{proof}
  Without loss of generality, we assume that $\edge_0 \in \M$. Then
  $\degree{\M}{v_0} > \degree{\N}{v_0}$ and $v_0$ is not \atight in \M. Since
  \trail is maximal and has even length, $\edge_{t-1} \in \N$. Therefore,
  $\degree{\M}{v_t} < \degree{\N}{v_t}$ and thus $v_t$ is not \btight in \M. By
  assumption, \trail has no \abfixed vertex. Therefore, $\trail \cap \M$ is
  internally 1-reconfigurable to $\trail \cap \N$ by
  Lemma~\ref{lemma:em1reconfig}.
\end{proof}

\begin{lemma}
  If \trail has odd length and $\edge_0$ is an \N-edge then
  \begin{enumerate}
	\item $\trail \cap \M$ is internally 1-reconfigurable to $\trail \cap \N$, and \label{itm:odd1reconfig}
	\item $\trail \cap \N$ is internally 2-reconfigurable to $\trail \cap \M$. \label{itm:odd2reconfig}
  \end{enumerate}
  \label{lemma:oddreconfig}
\end{lemma}
\begin{proof}
  We first prove part~\ref{itm:odd1reconfig} of the statement. Since \trail is
  maximal, the end nodes $v_0$ and $v_t$ are not \btight in \M.
  We recursively divide \trail into two subtrails $R$ and $S$, such that $R$
  has even length and $S$ has odd length and no end vertex of $S$ is \btight,
  either in $\M$ or in $\M - (R \cap \M) + (R \cap \N)$. The base case of this
  recursion that is not covered by Lemma~\ref{lemma:em1reconfig} consists of
  $S$ having a single \N-edge. Since none of the end vertices is \btight, the
  remaining \N-edge can be added without violating the degree constraints.

  Two subcases that occur when dividing \trail into subtrails $R$
  and $S$. First assume that there is some $i$, $0 < i < t$, such that $v_i$ is
  not \atight. If $i$ is even then we pick $R = v_i \uline{} \ldots \uline{}
  v_0$ and choose $S = v_i \uline{} \ldots \uline{} v_t$.  If $i$ is odd, we
  pick $R = v_i \uline{} \ldots \uline{} v_t$ and $S = v_0 \uline{} \ldots
  \uline{} v_i$.  Either way, $R \cap \M$ is internally 1-reconfigurable to $R
  \cap \N$ by Lemma~\ref{lemma:em1reconfig}. Furthermore, $S$ is an odd-length
  $(\N,\M)$-alternating and no end vertex of $S$ is \btight in $\M - (R \cap \M)
  + (R \cap \N)$. On the other hand, if there is no such $i$, then we pick $R =
  v_1 \uline{} \ldots \uline{} v_t$ and $S = v_0 \uline{\edge_0} v_1$. Now, $\M
  + \edge_0$ does not violate the degree constraints and $R \cap (\M +
  \edge_0)$ is internally 1-reconfigurable to $R \cap (\N - \edge_0)$ according
  to Lemma~\ref{lemma:em1reconfig}, since $v_1$ is not \atight in $\M +
  \edge_0$ and $v_t$ is not \btight by assumption. Therefore, in each case the
  subtrails are 1-reconfigurable as required.

  The proof of part~\ref{itm:odd2reconfig} is somewhat analogous:
  Since \trail is maximal, $v_0$ and $v_t$ cannot be \atight in \N. We again
  recursively divide \trail into two subtrails $R$ and $S$, such
  that $R$ has even length and $S$ has odd length and no end vertex of $S$ is
  \atight, in $\M$ or in $\M - (R \cap \M) + (R \cap \N)$. We distinguish the
  following two cases: First, we assume that each internal vertex of \trail is
  \btight. Then we pick $S = v_0 \uline{} v_1$ and $R = v_1 \uline{} \ldots
  \uline{} v_t$. We can reconfigure $S$ by removing $\edge_0$ and note that
  $R$ is internally 1-reconfigurable in $(\M - \edge_{t-1}) \symdiff \N$ by
  Lemma~\ref{lemma:em1reconfig}. Therefore, \trail is 2-reconfigurable.  Now,
  assume that there is some $i$, $0 < i < t$, such that $v_i$ is not \btight.
  If $i$ is even we choose $R = v_0 \uline{} \ldots \uline{} v_i$ and $S =
  v_{t} \uline{} \ldots \uline{} v_{i+1}$. If $i$ is odd we swap the choices of
  $R$ and $S$. Now, $R \cap \M$ is internally 1-reconfigurable to $R \cap \N$
  and no end vertex of $S$ is \atight, which means we can proceed recursively
  by reconfiguring $S$.
\end{proof}

\begin{lemma}
	If \trail is closed and has even length then 
  \begin{enumerate}
	\item $\trail \cap \M$ is internally 2-reconfigurable to $\trail \cap \N$ if~\trail is not alternatingly \abtight, and\label{itm:even2reconfig}
	\item $\trail \cap \M$ is internally 1-reconfigurable to $\trail \cap \N$ if \trail is neither \btight nor alternatingly \abtight.\label{itm:even1reconfig}
  \end{enumerate}
  \label{lemma:btightreconfig}
\end{lemma}
\begin{proof}
  If the conditions~\ref{itm:reconf:btight}--~\ref{itm:reconf:abtight} of
  Lemma~\ref{lemma:em1reconfig} are satisfied then we get that $\trail \cap \M$
  is internally 1-reconfigurable to $\trail \cap \N$. Thus, to
  complete the proof it is sufficient to consider the case that one of these
  conditions is violated. In order to prove 
  statement~\ref{itm:even2reconfig} assume that \trail is alternatingly \abtight, i.e., 
  \ref{itm:reconf:abtight} is violated. Then $\trail \cap \M$ is not internally
  $k$-reconfigurable to $\trail \cap \N$ for any $k \geq 1$, since no edges of
  $\trail \cap \M$ can be removed and no edges of $\trail \cap \N$ can be added
  without violating the degree constraints. If \trail is
  \btight, i.e., \ref{itm:reconf:btight} is violated, then we choose $R = v_1
  \uline{} \ldots \uline{} v_{t-1}$ and after removing the \M-edge $\edge_0$,
  $R \cap \M$ is internally 1-reconfigurable to $R \cap \N$ by
  Lemma~\ref{lemma:em1reconfig}. Therefore, $\trail \cap \M$ is internally
  2-reconfigurable to $\trail \cap \N$. Now, suppose that \trail  
  is \atight (\ref{itm:reconf:atight} is violated), then we can add the \N-edge
  $\edge_{t-1}$ and choose $R = v_0 \uline{} \ldots \uline{} v_{t-1}$. Then, by
  Lemma~\ref{lemma:em1reconfig}, $R \cap \M$ is internally 1-reconfigurable to
  $R \cap \N$, which completes the proof of statement~\ref{itm:even2reconfig}.

  In order to prove statement~\ref{itm:even1reconfig} it is sufficient to
  reconsider the case that \trail is \btight. If this is the case then no
  \N-edge can be added after removing any \M-edge and therefore, $\trail \cap
  \M$ is not internally 1-reconfigurable to $\trail \cap \N$.
\end{proof}

\subsection{External Alternating Trail Reconfiguration}

In the following we deal with even-length alternating cycles that are either
alternatingly \abtight or \btight. The two cases are somewhat special since we
will need to consider edges that are not part of the cycles themselves. Let
$\cycle = u_0 \uline{\edge_0} \ldots \uline{\edge_{t-1}} u_t$ be an
$(\M,\N)$-alternating cycle of even length in $\M \symdiff \N$.
We will first consider the case that \cycle is \btight. 2-reconfigurability of
\cycle is established by Lemma~\ref{lemma:btightreconfig}. We generalize the
approach from~\cite[Lemma 1]{Ito:11} to \abconstrained subgraphs to obtain a
characterization of 1-reconfigurable \btight even cycles in the case that \M
and \N are maximum and no vertex of $G$ is \abtight. The proof is analogous to
that of~\cite[Lemma 1]{Ito:11} and is not given here. We
denote by $\NotA(M)  :=   \{  v \in V(G) \mid 	 \text{ $v$ is not \atight in
\M}  \}$ and $\NotB(M)  :=   \{  v \in V(G) \mid 	 \text{ $v$ is not
\btight in \M} \}$ the sets of vertices that are not \atight in \M and
		not \btight in \M, respectively.
\begin{align*}
  &\Even(M) & :=   \{ && \!\!\!\!v \in \NotA(M) \mid & \text{ There is some even-length \M-alternating} \\
  &		    &      	  &&					& \text{ $vw$-trail starting with an \M-edge} \\
  &			&			&&					& \text{ s.t.~$w \in \NotB(M)$} \} \\
  &\NotB(G) & :=   \{ && v \in V(G) \mid & \text{ There is some maximum } \M \subseteq G \text{ satisfying} \\
  &			& 		&&					& \text{ the degree constraints s.t.~} v \in \NotB(\M) \} 
  \label{eqn:evenM}
\end{align*}

The following lemma is a generalization of~\cite[Lemma 2]{Ito:11} to the
\abconstrained subgraph setting.
\begin{lemma}
  If \M is maximum then $\Even(M) = \NotB(G)$.
  \label{lemma:evennotbtight}
\end{lemma}
\begin{proof}
  First, let $v \in \Even(M)$. Then there is some $vw$-alternating trail \trail
  such that $w \in \NotB(M)$. Then $v$ is not \atight and $w$ is not \btight in
  \M and therefore $\M' = \M \symdiff \trail$ is a maximum and satisfies the
  degree constraints. Since $v$ is not \btight in $\M'$ we have $v \in
  \NotB(G)$. Therefore, $\Even(M) \subseteq \NotB(G)$. In order to prove that
  $\NotB(G) \subseteq \Even(M)$ assume that $v \in \NotB(G)$. If $v \in
  \NotB(M)$ then it is also in $\NotB(G)$. Otherwise, let $\N \subseteq G$ a
  maximum degree-constrained subgraph such that $v \in \NotB(N)$. Then $v$ is
  \btight in \M, but not in \N. Suppose for a contradiction that there is no
  even-length $(\M,\N)$-alternating $vw$-trail such that $w \in \NotB(M)$. Then
  for any even-length $vw$-trail $w$ is \btight. But since \M and \N are
  maximum, in any maximal open $(\M,\N)$-alternating $vw$-trail $w$ is not
  \btight. Furthermore, \trail cannot be a cycle since $\degree{\M}{v} \neq
  \degree{\N}{v}$.
\end{proof}

\begin{lemma}
  If $G$ contains no \abfixed vertices, \M and \N are maximum, and \cycle is
  \btight then $\cycle \cap \M$ is 1-reconfigurable to $\cycle \cap \N$ if and
  only if there is a vertex $v$ of \cycle such that $v \in \Even(M)$.
  \label{lemma:btight1reconfig}
\end{lemma}
\begin{proof}
  We first show that if $v \in \cycle \cap \Even(M)$ then $\cycle
  \cap \M$ is 1-reconfigurable to $\cycle \cap \N$. If $v \in \cycle \cap
  \Even(M)$ then there is a \M-alternating $vw$-trail \trail of even length
  starting with an \M-edge such that $w$ is not \btight. In the order given by
  the trail there is an earliest edge
  \edge such that none of the successors of \edge in \trail are $C$-edges. We
  distinguish two cases: First, assume that \edge is an \M-edge. Without loss
  of generality, let $u_0$ be the \cycle-vertex incident to \edge. We obtain
  the \M-alternating subtrail $\trail' = u_0 \uline{} \ldots \uline{} w$ and
  $\M' = \M \symdiff \trail$ is maximum and satisfies the degree constraints.
  Further, since $u_0$ is \btight in \M and $w$ is not \btight \M, $\trail'
  \cap \M$ is 1-reconfigurable to $\trail' \cap \M'$ by
  Lemma~\ref{lemma:em1reconfig}. Then $u_0$ is not \btight in $\M'$ and $\cycle
  \cap \M$ is 1-reconfigurable to $\cycle \cap \M'$ by
  Lemma~\ref{lemma:em1reconfig}. Since $u_0$ is not \btight and $w$ is not
  \atight in $\cycle \cap \M'$, $\trail' \cap \M'$ is 1-reconfigurable to
  $\trail' \cap \M$ and thus we can undo the changes to \M caused by the
  reconfiguration on $\trail'$.

  In the second case we assume that \edge is a $(G-\M)$-edge.  Then there is a
  latest \cycle-edge on \trail such that none the successors of \edge in \trail
  are $C$-edges.  Without loss of generality, let $\edge = u_0 \uline{} u_1$
  such that we obtain an \M-alternating subtrail $\trail' = u_0 \uline{} u_1
  \uline{} \ldots \uline{} w$.  Then $u_0$ is not \atight in $\trail' \cap \M$,
  $w$ is not \btight in $\trail' \cap \M$, and therefore $\trail' \cap \M$ is
  internally 1-reconfigurable to $\trail' - \M$ by
  Lemma~\ref{lemma:em1reconfig}. In the resulting subgraph, $u_0$ is not
  \btight and $w$ is \btight. By using Lemma~\ref{lemma:em1reconfig} again, we
  can reconfigure the remaining parts of \cycle and undo the modifications to
  $\M - \cycle$ caused by the previous step, by considering the trail $w
  \uline{} \ldots \uline{} u_1 \uline{} u_2 \uline{} \ldots \uline{} u_t (=
  u_0)$. Hence, in both cases $\cycle \cap \M$ is 1-reconfigurable to $\cycle
  \cap \N$. Note that each time we invoke Lemma~\ref{lemma:em1reconfig}, we use
  the assumption that no vertex of $G$ is \abfixed.
  
  We now show that $v \in \cycle \cap \Even(M)$ is a necessary condition for
  $\cycle \cap \M$ to be 1-reconfigurable to $\cycle \cap \N$.  If no vertex of
  \cycle is in $\Even(M) = \NotB(G)$ (Lemma~\ref{lemma:evennotbtight}), then
  each vertex $v$ of \cycle is essentially \abfixed in the sense that no
  maximum \abconstrained subgraph exists such that $v$ is not \btight.
  Therefore $\cycle \cap \M$ cannot be 1-reconfigurable to $\cycle \cap \N$.
\end{proof}

We now characterize $k$-reconfigurable alternatingly \abtight cycles in $\M
\symdiff \N$ assuming that $G$ contains no \abfixed vertices. Such cycles
cannot occur in the matching reconfiguration setting since no vertex of a cycle
is \atight in this case. There is some conceptual similarity to the proof of
Lemma~\ref{lemma:btightreconfig}, but for the purpose of proving
Theorem~\ref{thm:stdcsconn} we cannot assume that \M and \N are both maximum.
Therefore, we cannot rely on Lemma~\ref{lemma:evennotbtight}, which simplifies
the problem of finding a maximum $ab$-constrained subgraph $\M'$ such that a
certain vertex is not \btight to checking for the existence of an alternating
trail. Instead, we now check if there is some \abconstrained subgraph $\M'$
such that the tightness of the \cycle-vertices in $\M'$ differs from their
tightness in \M.  The existence of a suitable $\M'$ can be checked in
polynomial time by constructing and solving suitable \dcs instances.

\begin{lemma}
  If \cycle is alternatingly \abtight and $G$ contains no \abfixed vertices
  then $\cycle \cap \M$ is $k$-reconfigurable to $\cycle \cap \N$ for any $k
  \geq 1$ if and only if there is some \abconstrained $\M' \subseteq G$ such
  that $\cycle \cap \M = \cycle \cap \M'$ and \cycle is not alternatingly
  \abtight in $\M'$.
  \label{lemma:altabreconfig}
\end{lemma}
\begin{proof}
  Let us first assume that there is some \abconstrained $\M' \subseteq G$ such
  that $\cycle \cap \M = \cycle \cap \M'$ and \cycle is not alternatingly
  \abtight in $\M'$. We show that $\cycle \cap \M$ is 1-reconfigurable to
  $\cycle \cap \N$.  Since \cycle is not alternatingly \abtight in $\M'$ there
  is some vertex $v$ of \cycle that is \btight in \M but not in $\M'$, or there
  is some vertex $u$ of \cycle that is \atight in \M but not in $\M'$. We will
  consider in detail the case that there is some $v$ of \cycle that is \btight
  in \M but not in $\M'$. The other case is analogous.  Since $\degree{\M}{v} <
  \degree{\M'}{v}$ there is an open $(\M,\M')$-alternating $vw$-trail \trail in
  $\M \symdiff \M'$ starting at $v$ such that \trail cannot be extended at $w$.
  Without loss of generality, we assume that $v = u_0$ and $u_0 \uline{} u_1$ is
  a $(\cycle - \M)$-edge.

  We consider two subcases: \trail has either even or odd length.  First
  assume that \trail is even. Let $R = w \uline{} \ldots \uline{} u=u_0
  \uline{} u_1 \uline{} \ldots \uline{} u_t$.  Since \trail is even, $v=u_t$ is
  not \atight in \M and $w$ is not \btight in \M. Furthermore, $R \symdiff \M$
  satisfies the degree constraints.  By Lemma~\ref{lemma:em1reconfig}, $R \cap
  \M$ is 1-reconfigurable to $R \cap (R \symdiff \M) = ((\cycle \cap \N) + (\trail \cap
  \M'))$. As a result, $\cycle \cap \M$ has been reconfigured to $\cycle \cap
  \N$. We now need to undo the changes in $\trail = R - \cycle$ caused by the
  previous reconfiguration.  Observe that $v$ is not \btight and $w$ is not
  \atight in $\trail \cap \M'$.  Therefore, we can invoke
  Lemma~\ref{lemma:em1reconfig} again to reconfigure $\trail \cap \M'$ to
  $\trail \cap \M$.
  In the second subcase we assume that \trail is odd and let $R = w \uline{}
  \ldots \uline{} u=u_0 \uline{} u_1 \uline{} \ldots \uline{} u_{t-1}$ and let
  $S = u_{t-1} \uline{} u_0 \uline{} \ldots \uline{} w$. Then $R$ and $S$ are
  open and have even length, $w$ is not \atight and $u_{t-1}$ is not \btight in
  \M. Therefore, by Lemma~\ref{lemma:em1reconfig}, $R \cap \M$ is
  1-reconfigurable to $R \cap (R \symdiff M)$. Now $u_{t-1}$ is not \atight and
  $w$ is not \btight in $R \symdiff \M$. Therefore we can use
  Lemma~\ref{lemma:em1reconfig} again in order to reconfigure $S \cap (R
  \symdiff \M)$ to $S \cap (S \symdiff (R \symdiff \M)) = S \cap (\M - e)$. As
  a result $\cycle \cap \M$ has been reconfigured to $\cycle \cap \N$.

  In order to prove the converse statement, assume that there is no \abconstrained $\M'
  \subseteq G$ such that $\cycle \cap \M = \cycle \cap \M'$ and \cycle is not
  alternatingly \abtight in $\M'$. Then each vertex of \cycle is essentially
  \abfixed.  Therefore, $\cycle \cap \M$ is not $k$-reconfigurable to $\cycle
  \cap \N$ for any $k \geq 1$. 
\end{proof}

\subsection{Reconfiguring \abconstrained Subgraphs}

For the overall task of deciding if \M is reconfigurable to \N we will
iteratively partition $\M \symdiff \N$ into alternating trails as shown in
Algorithm~\ref{alg:atd}. 
Given \M and \N such that $\M \symdiff \N$ is non-empty, the algorithm outputs
a decomposition of $\M \symdiff \N$ into trails $T_0,\ldots,T_{i-1}$ and a list
of \abconstrained subgraphs
$\M_0,\ldots,\M_{i-1}$ for some $i \geq 0$ such $\M = \M_0$ and $\N = \M_{i-1}$
such that for each $j$, $1 \leq j < i$, $\M_{j+1} = \M_j \symdiff \trail_j$.
In each iteration $i$ we need to find an \M-augmenting $uv$-trail in $\M_i
\symdiff \N$ such that $u$ and $v$ are not \btight in $\M_i$. Since  $\M_i$ is
\abconstrained  we can use for example the technique
from~\cite[Section 2]{Shiloach:81} to reduce the problem to obtaining an
alternating path in an auxiliary graph that is constructed from $\M \symdiff
\N$. This approach produces a suitable \M-augmenting trail if it exists in
polynomial time.  
Since the number of iterations performed by the algorithm is bounded by $|E(\M
\symdiff \N)|$, the overall running time is polynomial in the size of the input
graph. We are now ready to prove the main theorem. The structure of the proof
is somewhat similar to the proof that matching reconfiguration can be solved in
polynomial time, see~\cite[Proposition 2]{Ito:11}.

\begin{algorithm}
\caption{\atdecomp}
\label{alg:atd}
\DontPrintSemicolon
\SetKwInOut{Input}{input}
\SetKwInOut{InOut}{in/out}
\SetKwInOut{Output}{output}
\SetKwData{KLIST}{$K$}
\SetKwData{PP}{$P$}
\SetKwData{HH}{$H$}
\SetKwData{TT}{$T$}
\Input{\abconstrained $\M,\N \subseteq G$ s.t.~$\M \symdiff \N$ non-empty}
\Output{Lists of alternating trails and \abconstrained subgraphs}
$i \longleftarrow 0;\qquad \M_0 \longleftarrow \M$\;
\While{ $\M_i \neq \N$ }
{
  Find \M-augmenting $uv$-trail \TT in $\M_i \symdiff \N$ s.t.~$u$ and $v$ are not \btight in $\M_i$\;
  \If{such \TT does not exist} {
	Let \TT be any maximal $(\M_i,\N)$-alternating trail in $\M_i \symdiff \N$\;
  }
  $\TT_i \longleftarrow \TT$\;
  $\M_{i+1} \longleftarrow \M_i \symdiff \TT$\;
  $i \longleftarrow i + 1$\; 
}
\Return $[\M_0,\ldots,\M_{i-1}], [\TT_0,\ldots,\TT_{i-1}]$\;
\end{algorithm}

\begin{theorem}
  \stdcsconn can be solved in polynomial time.
  \label{thm:stdcsconn}
\end{theorem}

\begin{proof}
  Let $\instance = (G', M', N', a', b',k)$ be a \stdcsconn instance. Let
  $F \subseteq G'$ be the $\M'$-fixed subgraph produced by
  Algorithm~\ref{alg:mfixed}.  If $(\M' \symdiff \N') \cap F$ is non-empty then
  some $\M'$-fixed edge needs to be reconfigured, which is impossible by
  Proposition~\ref{prop:mfixed}. Otherwise, we consider the subinstance 
  $\instance_{G'-F} = (G, M, N, a, b,k)$ (see Proposition~\ref{prop:subinstance}).  
  
  Without loss of generality we assume that $|E(\M)| \leq |E(\N)|$.  We process the
  alternating trails in $\M \symdiff \N$ output by Algorithm~\ref{alg:atd} one
  by one in the given order. During the process, we observe the following
  types of $(\M,\N)$-alternating trails: i) even-length trails that are not
  \btight or alternatingly \abtight cycles, ii) \M-augmenting trails, iii)
  \N-augmenting trails, iv) \btight even-length cycles, and v) alternatingly
  \abtight even-length cycles. Since $|E(\M)| \leq |E(\N)|$ there are at least as
  many type ii)-trails as type iii)-trails. Note that each condition in the
  lemmas~\ref{lemma:em1reconfig}--~\ref{lemma:altabreconfig} can be checked in
  polynomial time.
  We distinguish the following cases:
  \paragraph{Case $k \geq 2$.}
	By construction, in each step $i$, if $\trail_i$ is in categories i)--iv)
	then $\trail \cap \M_i$ is 2-reconfigurable to $\trail \cap \M_{i+1}$ by
	Lemmas~\ref{lemma:even1reconfig},~\ref{lemma:oddreconfig},
	and~\ref{lemma:btightreconfig}. If $\trail_i$ is an even-length
	alternatingly \abtight cycle (type v)) then Lemma~\ref{lemma:altabreconfig}
	gives necessary and sufficient conditions under which $\trail_i \cap \M_i$
	is $k$-reconfigurable to $\trail_i \cap \M_{i+1}$ for any $k \geq 1$. These
	conditions can be checked in polynomial time and do not depend on what
	edges are present in $\M_i$ outside of $\trail_i$ The reconfigurability of
	$\trail_i$ is a property solely of $\trail_i$, $G$, and the degree bounds. 
  \paragraph{Case $k = 1$ and $|E(\M)| < |E(\N)|$.}
	Trails of types i) and ii) are 1-reconfigurable by
	Lemmas~\ref{lemma:even1reconfig} and~\ref{lemma:oddreconfig}. Due to the
	preference given to \M-augmenting trails in Algorithm~\ref{alg:atd}, if
	$\trail_i$ is an \N-augmenting trail or a \btight cycle then $|E(\M_i)| \geq
	|E(\N)|$. Therefore by Lemma~\ref{lemma:oddreconfig}
	or~\ref{lemma:btight1reconfig} we have that $\trail_i \cap \M_i$ is
	2-reconfigurable to $\M_{i+1}$, but no intermediate subgraph is of size
	less than $|E(\M)|-1$. The alternatingly \abtight cycles can be dealt with
	just as in the previous case. 
  \paragraph{Case $k = 1$ and $|E(\M)| = |E(\N)|$, both not maximum.} 
	In this case we increase the size of $\N$ by one, using an \N-augmenting
	\trail, to obtain $\N'$. We first reconfigure \M to $\N'$ as in the case
	before. If successful, the result $\N'$ is 2-reconfigurable to $\N$ by
	Lemma~\ref{lemma:oddreconfig}. No intermediate subgraph is of size less
	than $|E(\M)|-1$.
  \paragraph{Case $k = 1$, $|E(\M)| = |E(\N)|$, both maximum.} 
	Since \M and \N are
	maximum, each trail $\trail_i$ is of type i), iv), or v). Therefore,
	each open trail is 1-reconfigurable by Lemma~\ref{lemma:even1reconfig}. We
	need to check the 1-reconfigurability of each cycle according to
	lemmas~\ref{lemma:btight1reconfig} (type iv)) and~\ref{lemma:altabreconfig} (type v)).
\end{proof}

The running time of the decision procedure is dominated by the time needed to
check the conditions of lemmas~\ref{lemma:btight1reconfig}
and~\ref{lemma:altabreconfig}. Overall, this amounts to solving
$O(|V(G)|\cdot|E(G)|^2)$ \dcs instances, which takes time
$O(|E(G)|^{\frac{3}{2}})$ per instance using the algorithm
from~\cite{Gabow:83}. 

\bibliographystyle{plain}
\bibliography{paper}

\end{document}